\documentclass[mathfont=ptmx]{article}

\usepackage{arxiv}

\usepackage[utf8]{inputenc} 
\usepackage[T1]{fontenc}    
\usepackage{hyperref}       
\usepackage{url}            
\usepackage{booktabs}       
\usepackage{amsfonts}       
\usepackage{nicefrac}       
\usepackage{microtype}      
\usepackage{lipsum}		
\usepackage{graphicx}
\usepackage{natbib}
\usepackage{doi}

\usepackage{xcolor}

\usepackage{amsmath}
\usepackage{amssymb}
\usepackage{amsthm}

\usepackage{centernot}
\usepackage{enumitem}
\usepackage{multirow}
\usepackage{booktabs}
\usepackage{mathtools}
\usepackage{booktabs}
\usepackage{algorithm}
\usepackage{algorithmic}
\usepackage{colortbl}
\newtheorem{remark}{Remark}

\newtheorem{proposition}{Proposition}
\newtheorem{corollary}{Corollary}
\newtheorem{lemma}{Lemma}
\newtheorem{definition}{Definition}
\newtheorem{example}{Example}
\newtheorem{theorem}{Theorem}

\title{Fair Allocation with Special Externalities}

\author{{Shaily Mishra} \\
	International Institute of\\ Information Technology (IIIT)\\
	Hyderabad, India \\
	\texttt{shaily.mishra@research.iiit.ac.in} \\
	\And
	{Manisha Padala} \\
	International Institute of\\ Information Technology (IIIT)\\
	Hyderabad, India \\
	\texttt{manisha.padala@research.iiit.ac.in} \\
	\And
	Sujit Gujar \\
	International Institute of\\ Information Technology (IIIT)\\
	Hyderabad, India \\
	\texttt{sujit.gujar@iiit.ac.in} \\
}

\begin{document}
\maketitle

\begin{abstract}
Most of the existing algorithms for fair division do not consider externalities. Under externalities, the utility an agent obtains depends not only on its allocation but also on the allocation of other agents. An agent has a positive (negative) value for the assigned goods (chores). This work focuses on a special case of externality, i.e., an agent receives positive or negative value for unassigned items independent of which other agent gets it. We show that it is possible to adapt existing algorithms using a transformation to ensure certain fairness and efficiency notions in this setting. Despite the positive results, fairness notions like proportionality need to be re-defined. Further, we prove that maximin share (MMS) may not have any multiplicative approximation in this setting. Studying this domain is a stepping stone towards full externalities where ensuring fairness is much more challenging.
\end{abstract}

\keywords{Resource Allocation \and Fairness \and Externalities}

\section{Introduction}
\label{sec:intro}
We consider the problem of allocating $m$ indivisible items among $n$ agents who report their valuations for the items. The objective is to ensure fair allocation for a desirable notion of fairness. These scenarios often arise in the division of inheritance among family members, divorce settlements and distribution of tasks among workers~\cite{brams_taylor_1996,moulinfairness,Segal_Halevi_2019,propSTEIHAUS,webbcakecutting}. Economists have proposed many fairness and efficiency notions widely applicable in such real-world settings. Researchers also explore the computational aspects of some widely accepted fairness notions \cite{Caragiannis2016,barman2018finding,computationallyguarantees1,EQ1freemangoods,procaccia2014mms}. Such endeavours have led to web-based applications like Spliddit \footnote{\url{www.spliddit.org}}, The Fair Proposals System\footnote{\url{www.fairproposals.com}}, Coursematch  \footnote{\url{www.coursematch.io}}, Divide Your Rent Fairly \footnote{\url{www.nytimes.com/interactive/2014/science/rent-division-calculator.html}}, etc.
However, most approaches do not consider agents with \emph{externalities}, which we believe is restrictive.


In general, externality implies that the agent's utility depends not only on their bundle but also on the bundles allocated to other agents. Such a scenario is relatively common, mainly in allocating necessary commodities. For example, the COVID-19 pandemic resulted in a sudden and steep requirement for life-supporting resources like hospital beds, ventilators, and vaccines. Getting a vaccination affects an agent positively. Even if someone else gets vaccinated instead of the agent, the agent values it positively, possibly less. However, not receiving a ventilator in time results in negative utility for the patient and family. Such a complex valuation structure is modeled via externalities.

Generally with externalities, the utility of not receiving an item depends on which other agent receives it. That is, each agent's valuation for an item is an $n$-dimensional vector. The $j^{th}$ component corresponds to the value an agent obtains if the item is allocated to agent $j$.
In this work, we consider a special case of externalities in which the agents incur a cost/benefit for not receiving an item. Yet, the cost/benefit is \emph{independent} of which other agent receives the item. This setting is referred to as 2-D, i.e., value $v$ for receiving an item and $v'$ otherwise. When there are only two agents, the 2-D domain is equivalent to the domain with general externalities. We refer to the agent valuations in the absence of externalities as 1-D. For the 2-D domain, we consider both goods/chores with positive/negative externality for the following fairness notions.
 

\smallskip
\noindent\textbf{Fairness Notions.}
Envy-freeness (EF) is the most common notion of fairness. It ensures that no agent has higher utility for other agent's allocation \cite{Foley1967ResourceAA}. Consider two agents - $\{1,2\}$ and two goods - $\{g_1,g_2\}$; agent 1 has a value of 6 for good $g_1$ and 5 for good $g_2$, while agent 2 values $g_1$ at 5 and $g_2$ at 6. Then allocating $g_1$ and $g_2$ to agent 1 and 2, respectively, is EF. Such an allocation will not be EF if we consider externalities. For example, if agent $1$ receives a negative utility of $-1$ and $-100$ for not receiving $g_1$ and $g_2$, respectively. And agent $2$  receives a negative utility of $-100$ and $-1$ for not receiving $g_1$ and $g_2$, respectively.

Externalities introduce complexity, so much that the definition of proportionality cannot be adapted to the 2-D domain. Proportionality (PROP) ensures that every agent receives at least $1/n$ of its complete bundle value \cite{propSTEIHAUS}. In the above example, each agent should receive goods worth at least $11/2$. Guaranteeing this amount is impossible in 2-D, as it does not consider the dis-utility of not receiving goods. 
Moreover, it is known that EF implies PROP in the presence of additive valuations. However, in the case of 2-D, it need not be true, i.e., assigning $g_2$ to agent 1 and $g_1$ to agent 2 is EF but not PROP.

Finally, we consider a relaxation of PROP, the maximin share (MMS) allocation. Imagine asking an agent to divide the items into $n$ bundles and take the minimum valued bundle. The agent would divide the bundles to maximise the minimum utility, which is the MMS share of the agent. An MMS allocation guarantees every agent its MMS share. Even for 1-D valuations, MMS allocation may not exist; hence researchers find multiplicative approximation $\alpha$-MMS. An $\alpha$-MMS allocation guarantees at least $\alpha$ fraction of MMS share to every agent. Authors in \cite{GargMMS2020} provides an algorithm that guarantees $3/4+1/12n$-MMS for goods and authors in \cite{Xinhuangmmschores} guarantees $11/9$-MMS for chores. 
In contrast, we prove that for 2-D valuation, it is impossible to guarantee multiplicative approximation to MMS. Thus, in order to guarantee existence results, we propose relaxed multiplicative approximation and also explore additive approximations of MMS guarantees.

In general, it is challenging to ensure fairness in the settings with full externality, hence the special case of 2-D proves promising. Moreover, in real-world applications, the 2-D valuations helps model various situations (e.g., COVID-19 resource allocation mentioned above).

\paragraph{Our Approach.}
There is extensive literature available for fair allocations, and we primarily focus on leveraging existing algorithms to 2-D.
Towards guaranteeing fairness notion in 2-D, we propose a property preserving transformation $\mathfrak{T}$ that converts 2-D valuations to 1-D; i.e., an allocation that satisfies a property in 2-D also satisfies it in transformed 1-D and vice-versa. Moreover, the 1-D valuations obtained via $\mathfrak{T}$ satisfy the assumptions required to apply the existing algorithms for finding fair allocations. Along with fairness, typically certain efficiency notions are also considered. Hence, we also study if our transformation retains the efficiency notions.

\smallskip
\noindent\textbf{Contributions.}
\begin{enumerate}
    \item We propose $\mathfrak{T}$ that retains fairness notions such as EF, MMS, and its additive relaxations and efficiency notions such as MUW and PO (Theorem~\ref{thm:transformation}). Thus, we can adapt the existing algorithms for the same.
    \item We introduce PROP-E for general valuations in the presence of full externalities (Section~\ref{sec:prelim}) and derive relation with existing proportionality extensions (Section~\ref{sec:propE}).
    \item We prove that multiplicative approximation of MMS may not exist in 2-D (Theorem~\ref{thm:mulmmsdoesntexist}).
    \item We propose Shifted $\alpha$-MMS which is a novel way of approximating MMS in 2-D  (Section~\ref{sec:shiftedalphamms}).
\end{enumerate}

\subsection*{Related Work}
While fair resource division has an extremely rich literature, externalities in fair division is less explored.
\citet{velez2016fairness} extend the notion of EF in externalities and explored EF allocation of indivisible goods and money among interested agents in presence of externalities.
\citet{externalitiesdivisbleijcai13} generalize PROP and EF for divisible goods allocation with positive externalities. \citet{treibich2019welfare} study egalitarian social welfare in presence of average externalities for divisible goods.  Further, \citet{externalitiesghodsi} propose average-share definition, an extension of PROP, and study (EMMS) MMS allocation for indivisible goods with positive externalities. Note that in our setting MMS share is equivalent to EMS. Authors in \cite{aziz2021fairnessexternalities} explore EF1/EFX for the specific setting of two and three agents. For two agents, their setting is equivalent to 2-D, hence existing algorithms for EF, PROP and their additive relaxations \cite{Caragiannis2016,EFXsolution,PROPXAziz2019} suffice. Beyond two agents, the setting is more general and \citet{aziz2021fairnessexternalities} prove the non-existence of EFX for three agents. In contrast, for the special case of 2-D, EFX always exists for three agents since it exists in 1-D~\cite{chaudhury2020efx}. Further \cite{aziz2021fairnessexternalities}, provide extension of PROP for additive valuations with full externalities.

We breifly summarize the existing algorithms for 1-D valuations available for each of the fairness notions.

\smallskip
\noindent\textit{Envy-freeness.} EF may not exist for indivisible items. Hence we consider two prominent relaxations of EF, Envy-freeness up to one item (EF1) \cite{Budish11,Lipton2004} and Envy-freeness up to any item(EFX) \cite{Caragiannis2016}. We have poly-time algorithms to find EF1 in general monotone valuations for goods \cite{Lipton2004} and chores \cite{vaishenvycyle}. For additive valuations, EF1 can be found using Round Robin \cite{Caragiannis2016} in goods or chores, and Double Round Robin \cite{azizdoublerrijcai19} in combination. 
Authors in \cite{EFXsolution} present an algorithm to find EFX allocation under identical general valuations for goods. 
Researchers have also studied fair division in presence of strategic agents \cite{barmanaamas19strategicagents,bei2020truthful,padala2021mechanism}.


\smallskip
\noindent\textit{Proportionality.} PROP1 and PROPX are popular relaxation of PROP. For additive valuations, EF1 implies PROP1, and EFX implies PROPX.  Unfortunately, in paper \cite{PROPXAziz2019}, the authors showed the PROPX for goods may not always exists. Authors in \cite{li2021almost} explored (weighted) PROPX showed that a (weighted) PROPX allocation
always exists and can be computed efficiently.

\smallskip
\noindent\textit{MMS.} MMS allocations do not always exist \cite{procaccia2014mms,kurokawa2016}. The papers \cite{procaccia2014mms,Amanatidismms2017,Barman2018GroupwiseMF,garg2018} showed that 2/3-MMS for goods always exists. Paper \cite{ghodsi2017mms,GargMMS2020} showed that 3/4-MMS for goods always exists. Authors in \cite{GargMMS2020} provides an algorithm that guarantees $3/4+1/12n$-MMS for goods. Authors in \cite{azizmmschores} presented a polynomial-time algorithm for 2-MMS for chores. The algorithm presented in \cite{Barman2018GroupwiseMF} gives 4/3-MMS for chores. Authors in \cite{Xinhuangmmschores} showed that 11/9-MMS for chores always exists. Authors in \cite{kulkarni2021mixedmanna} explored $\alpha-$MMS for a combination of goods and chores.

\smallskip
 \noindent\textit{Fair and Efficient. } In \cite{Caragiannis2016}, the authors showed that MNW allocation is EF1 and PO for indivisible goods and \cite{barman2018finding} gave a pseudo-polynomial time algorithm. For a combination of resources, the authors in \cite{azizdoublerrijcai19} presented a polynomial-time algorithm to find EF1 and PO for two agents. An Algorithm to find PROP1 and fractional PO which is 
stronger than PO was proposed by \cite{PROPXAziz2019} for a combination of resources. Authors in \cite{azizef1um} proposed a pseudo-polynomial time for finding utilitarian maximizing among EF1 or PROP1 in goods.



\smallskip

\section{Preliminaries}
\label{sec:prelim}

We consider a resource allocation problem $(N,M,\mathcal{V})$ for  determining an allocation $A$ of $M = [m]$ indivisible items among $N=[n]$ interested agents, $m,n \in \mathbb{N}$. We only allow complete allocation and no two agents can receive the same item. That is, $A=(A_1, \ldots, A_n)$, $A\in {N^M}$ s.t., $\forall i, j \in N$, $i \neq j; A_i \cap A_j = \emptyset$ and $\bigcup_i A_i = M$. We denote the allocation for all the agents except $i$ as $A_{-i}$.

\smallskip
\noindent \textbf{2-D Valuations.} The valuation function for $n$ agents is denoted by $\mathcal{V}=\{V_1,V_2,\ldots,V_n\}$. For each $i\in N$, $V_i:2^M \rightarrow \mathbb{R}^2$, i.e., $V_i \in \mathbb{R}^{2^{2^M}}$. For any bundle $S\subseteq M$, $V_i(S) = (v_i(S),v'_i(S))$, where $v_i(S)$ denotes the value for receiving bundle $S$ and $v'_i(S)$ for not receiving $S$. The value an agent $i$ has for item $k$ in 2-D is given by $(v_{ik}, v'_{ik})$. If $k$ is a good (chore), then $v_{ik} \geq 0 \ (v_{ik} \leq 0)$. For positive (negative) externality $v'_{ik} \ge 0 \ (v'_{ik} \le 0)$.

The utility an agent $i \in N$ obtains for a bundle $S \subseteq M$ is,
$$u_i(S) = v_i(S) + v'_i(M \setminus S)$$  Also, $u_i(\emptyset) = 0 + v'_i(M)$ and utilties in 2-D are not normalized\footnote{Utility is normalized when $u_i(\emptyset) = 0, \forall i$}. When agents have additive valuations, $u_i(S) = \sum_{k\in S} v_{ik} +$  $\sum_{k \notin S} v'_{ik}$. 


We assume monotonicity of utility for goods, i.e., $\forall S \subseteq T \subseteq M$, $u_i(S) \le u_i(T)$ and anti-monotonicity of utility for chores, i.e, $u_i(S) \ge u_i(T)$. We use the term \emph{full externalities} to represent complete externalities, i.e., each agent has $n$-dimensional vector for its valuation for an item.

Given the notations, we next define fairness notions considered in this paper.

\paragraph{Important Definitions.}
Since Envy-freeness (EF) may not exist for indivisible items, we consider EF1 and EFX. For goods, an allocation is EF1 when the agent values its own bundle no less than it values any other agent's bundle with the \textit{most} valued item removed. EFX is stronger than EF1 and requires that the agent values its own bundle no less than the other agent's bundle with the \textit{least} valued item removed. For chores, similar definition applies but unlike in goods, a chore is removed from the agent's own bundle and then compared with the other agents'. A common definition for is as follows,
\begin{definition}[Envy-free (EF) and relaxations \cite{azizdoublerrijcai19,Budish11,Caragiannis2016,Foley1967ResourceAA,velez2016fairness}]
For the items (chores or goods) an allocation $A$ that satisfies $\forall i,j \in N$,
\begin{align}
   &  u_i(A_i) \ge u_i(A_j)  \mbox{ is EF} \nonumber\\
    & \begin{rcases}
     v_{ik} < 0,   u_i(A_i \setminus \{ k\}) \ge u_i(A_j);\forall k \in A_i\\ 
v_{ik} > 0,   u_i(A_i ) \ge u_i(A_j \setminus \{ k\});\forall k \in A_j \ 
    \end{rcases}  \mbox{ is EFX}  \nonumber\\
   &  u_i(A_i \setminus \{ k\}) \ge u_i(A_j \setminus \{k\});\exists k \in \{A_i \cup A_j\}  \mbox{ is EF1} \nonumber
\end{align}
\end{definition}

\noindent Note that beyond 2-D, one must include the concept of swapping bundles to generalize the above definition as in~\cite{velez2016fairness,aziz2021fairnessexternalities}. We next state the definition of proportionality for 1-D.

\begin{definition}[Proportionality (PROP) \cite{propSTEIHAUS}]\label{def:prop}
An allocation $A$ is said to be proportional, if $\forall i \in N$, $u_i(A_i) \ge \frac{1}{n} \cdot u_i(M)$.
\end{definition}
For 2-D, achieving PROP is impossible as discussed in Section \ref{sec:intro}. To capture proportional allocations under externalities, we now introduce \textit{Proportionality with externality} (PROP-E). Informally, while PROP guarantees $1/n$ share of the entire bundle, PROP-E guarantees $1/n$ share of the sum of utilities for all bundles. Note that, PROP-E is not limited to 2-D and applies to a general externality setting. 
Formally,


\begin{definition}[Proportionality with externality (PROP-E)]
\label{def:nprop}
An allocation $A$ satisfies PROP-E if, $\forall i \in N$,
\begin{align}
 u_i(A_i) &\geq \frac{1}{n}\cdot\sum_{j\in N} u_i(A_j)
\end{align}
\end{definition}

\noindent Analogous to EFX/EF1, we now define the relaxations for PROP-E for the combination of goods and chores,
\begin{definition}[PROP-E relaxations]
\label{def:propX1-E}
An allocation $A$ $\forall i, \forall j \in N$, satisfies  PROPX-E if it is PROP-E up to any item, i.e.,
\begin{align}
    & \begin{rcases}
v_{ik}>0,      u_i(A_i \cup \{k\}) \ge \frac{1}{n}\sum_{j\in N} u_i(A_j);\forall \ k \in \{M \setminus A_i\} \\
v_{ik}<0, u_i(A_i \setminus \{k\}) \ge   \frac{1}{n}\sum_{j\in N} u_i(A_j);\forall \ k \in A_i
    \end{rcases}  \nonumber\\ 
    & \mbox{Next,  $A$ satisfies PROP1-E if it is PROP-E up to an item, i.e.,} \nonumber\\
     & \begin{rcases}
u_i(A_i \cup \{k\}) \ge \frac{1}{n}\sum_{j\in N} u_i(A_j); \exists \ k \in \{M \setminus A_i\}  \mbox{ or,}\\ 
u_i(A_i \setminus \{k\}) \ge \frac{1}{n}\sum_{j\in N} u_i(A_j); \exists \ k \in A_i    \nonumber\end{rcases}
\end{align}
\end{definition}

\noindent Finally, we state the definition of MMS and its multiplicative approximation.
\begin{definition}[Maxmin Share MMS \cite{Budish11}]
\label{def:mms}
An allocation $A$ is said to be MMS if $\forall i \in N, u_i(A_i) \ge \mu_i$, where
$$\mu_i = \max_{(A_1,A_2,\ldots,A_n) \in \prod_n(M)} \min_{j \in N} u_i(A_j)$$

An allocation $A$ is said to be $\alpha$-MMS if it guarantees $u_i(A_i) \ge \alpha \cdot \mu_i $ for $\mu_i \ge 0$, where  $\alpha \in [0,1]$  and $u_i(A_i) \ge \frac{1}{\alpha} \cdot \mu_i $ when $\mu_i \le 0$, where  $1/\alpha \ge 1$ and $\alpha > 0$.
\end{definition}

\noindent Since it is common to consider efficiency with notions, we next define Pareto-optimality, a popular efficiency notions. 
\begin{definition}[Pareto-Optimal (PO)] An allocation $A$ is PO if $\; \nexists \;  A'$ s.t.,
$ \forall i \in N$, $u_i(A'_i) \ge u_i(A_i)$ and $ \exists i \in N$, $u_i(A'_i) > u_i(A_i) $.
\end{definition}
\noindent We also consider efficiency notions like Maximum Utilitarian Welfare (MUW), that maximizes the sum of agent utilities. Likewise Maximium Nash Welfare (MNW) maximizes the product of agent utilties and Maximum Egalitarian Welfare (MEW) maximizes the minimum agent utility.

In the next section, we define a transformation from 2-D to 1-D that plays a major role in adaptation of existing algorithms for ensuring desirable properties.

\section{Reduction from 2-D to 1-D}
\label{sec:reductiontu}
We define a transformation $\mathfrak{T}: \mathcal{V} \rightarrow \mathcal{W}$, where $\mathcal{V}$ is the valuations in 2-D, i.e., $\mathcal{V}=\{V_1,V_2,\ldots,V_n\}$ and $\mathcal{W}$ is the valuations in 1-D, i.e., $\mathcal{W}=\{w_1,w_2,\ldots,w_n\}$. Note that $w_i : 2^M \rightarrow \mathbb{R}$. The transformation $\mathfrak{T}$ reduces $V_i \in \mathbb{R}^{2^{2^M}}$ to $w_i \in \mathbb{R}^{{2^M}}$.

\begin{definition}[Transformation $\mathfrak{T}$]
\label{def:Tu}
Given a resource allocation problem $(N, M, \mathcal{V})$ we obtain the corresponding 1-D valuations denoted by $\mathcal{W}=\mathfrak{T}\left(\mathcal{V}(\cdot)\right)$ as follows, $\forall i\in N$
\begin{equation}
        \label{eq:transformTus}
    w_i(A_i) = \mathfrak{T} (V_i(A_i)) = v_i(A_i) + v'_i(A_{-i}) - v'_i(M)
\end{equation}
When valuations are additive, we obtain $$w_i(A_i) = v_i(A_i) - v'_i(A_i)$$
\end{definition}

\smallskip
\noindent\textit{Example.} Consider two goods $\{g_1,g_2\}$ and an agent with 2-D additive valuations given as: $\{g_1 : (6, -1) $, $g_2 : (5,-100)\}$. We apply $\mathfrak{T}$ and obtain $w_1(\{g_1\}) = 7$ and $w_1(\{g_2\}) = 105$.

For an allocation $A$, the utility obtained by an agent in 2-D is  $u_i(A_i)$ and the corresponding utility in 1-D is $w_i(A_i)$. Note that utility is equal to the valuation in 1-D. 
\begin{lemma}
\label{lemma:winormalized} For goods (chores), under monotonicity (anti-monotonicity) of $\mathcal{V}$,
$\mathcal{W}=\mathfrak{T}\left(\mathcal{V}(\cdot)\right)$ is normalized, monotonic (anti-monotonic), and non-negative (negative).
\end{lemma}
\begin{proof}
We assume monotonicity of utility for goods in 2-D, i.e., $\forall i \in N, u_i(\cdot)$ is monotone. Therefore, for an $S\subseteq M, w_i(S) = u_i(S) - v'_i(M)$ is also monotone.
Further, $w_i(\emptyset)=v_i(\emptyset) + v'_i(M)-v'_i(M) = 0$ is normalized. Since $w_i(\cdot)$ is monotone and normalized, it is non negative for goods. Similarly we can prove that $w_i(\cdot)$ is normalized, anti-monotonic and non-negative for chores.
\end{proof}
 

\begin{theorem}
\label{thm:transformation}
An Allocation $A$ is $\mathfrak{F}$-Fair and $\mathfrak{E}$-Efficient in $\mathcal{V}$ \emph{iff} $A$ is $\mathfrak{F}$-Fair and $\mathfrak{E}$-Efficient in the transformed 1-D $\mathcal{W}$, $\mathfrak{F} \in \{EF, EF1, EFX, PROP-E, PROP1-E, PROPX-E, MMS\}$ and $\mathfrak{E} \in \{PO, MUW\}$.
\end{theorem}
\begin{proof}
We first consider $\mathfrak{F} =$ EF. Let allocation $A$ be EF in $\mathcal{W}$ then,
\begin{equation*}
\begin{aligned}
\forall i,\  \forall j,\  w_i(A_i) &\ge w_i(A_j)
\\
v_i(A_i) + v'_i(A_{-i}) - v'_i(M) &\ge v_i(A_j) + v'_i(A_{-j})  -  v'_i(M)
\\
u_i(A_i)  &\ge u_i(A_j)
\end{aligned}
\end{equation*}
Starting with EF allocation in 2-D, we can prove it is EF in 1-D similarly. The complete proof of Theorem~\ref{thm:transformation} for other notions is provided in Appendix Section~\ref{sec:completeproofoftheorem1}.
\end{proof}

\noindent From Lemma~\ref{lemma:winormalized} and Theorem~\ref{thm:transformation}, we obtain the following.
\begin{corollary}
To determine \{EF, EF1, EFX, MMS\} fairness and \{PO, MUW\} efficiency, we can apply existing algorithms to the transformed $\mathcal{W}=\mathfrak{T}\left(V(\cdot)\right)$ for general valuations. 
\end{corollary}



Note that applying any algorithm on the 2-D utility values directly without transformation may not work. We state few examples are below. Modified \textit{leximin} algorithm to find PROP1 and PO for chores for 3 or 4 agents given in \cite{chen2020fairness} does not find PROP1-E (or PROP1) and PO in 2-D when applied on utilities. The following example demonstrates the same,
\begin{example}
\emph{Consider 3 agents $\{1,2,3\}$ and 4 chores $\{c_1,c_2,c_3,c_4\}$ with positive externality. The 2-D valuation profile is as follows, $V_{1c_1}=(-30,1)$, $V_{1c_2}=(-20,1)$, $V_{1c_3}=(-30,1)$, $V_{1c_4}=(-30,1)$, $V_{3c_1}=(-1,40)$, $V_{3c_2}=(-1,40)$, $V_{3c_3}=(-1,40)$, and $V_{3c_4}=(-1,40)$. The valuation profile of agent $2$ is the same as that of agent $1$. Allocation $\{\emptyset,\emptyset,(c_1,c_2,c_3,c_4)\}$ is the leximin allocation, which is not PROP1-E. However, allocation $\{c_3,(c_2,c_4),(c_1)\}$ is leximin allocation on transformed valuations; it is PROP1 and PO in $\mathcal{W}$ and it is PROP1-E and PO in $\mathcal{V}$}.
\end{example}
In the same way, for chores, paper \cite{li2021almost} showed that any PROPX allocation ensures 2-MMS for symmetric agents doesn't extend to 2-D. For example, consider two agents $\{1,2\}$  having additive identical valuations for six chores $\{c_1,c_2,c_3,c_4,c_5,c_6\}$, given as $V_{1_{c1}}=(-9,1)$, $V_{1_{c2}}=(-11,1)$, $V_{1_{c3}}=(-12,1)$, $V_{1_{c4}}=(-13,1)$, $V_{1_{c5}}=(-9,1)$, and $V_{1_{c6}}=(-1,38)$. Allocation $A=\{(c_1,c_2,c_3,c_4,c_5),(c_6)\}$ is PROPX-E, but is not 2-MMS in $\mathcal{V}$. While it holds true for $\mathcal{W}$.

Further, adapting certain fairness or efficiency criteria to 2-D is not straightforward. E.g., MNW cannot be defined in 2-D because agents can have positive or negative utilities. Hence, certain results from 1-D like MNW implies EF1 and PO for additive valuations \cite{Caragiannis2016} do not apply for 2-D. The authors proved that MNW allocation gives at least $\frac{2}{1+\sqrt{4n-3}}$-MMS value to each agent in paper\cite{Caragiannis2016}, which doesn't imply for 2-D. Similarly, we show that approximation to MMS, $\alpha$-MMS, does not exist in the presence of externalities (see Section \ref{sec:alphamms}).


\section{Proportionality in 2-D}
\label{sec:propE}
We remark that ensuring PROP (Def. \ref{def:prop}) is too strict in 2-D. As a result, we introduce PROP-E and its additive relaxations in Defs. \ref{def:nprop} and \ref{def:propX1-E} for general valuations. 
\begin{proposition} For additive 2-D,
We can adapt the existing algorithms of PROP and its relaxations to 2-D using $\mathfrak{T}$.
\end{proposition}
\begin{proof} In the absence of externalities, for additive valuations, PROP-E is equivalent to PROP as $\forall i, \sum_{j \in N} u_i(A_j) = v_i(M)$. From Theorem~\ref{thm:transformation}, we know that $\mathfrak{T}$ retains PROP-E and its relaxations, and hence all existing algorithms of 1-D is applicable using $\mathfrak{T}$.
\end{proof}

It is known that EF $\implies$ PROP for sub-additive valuation in 1-D. 
As formally presented in Corollary~\ref{corr:1}, in the case of PROP-E also, $\forall i,j \in N$, $u_i(A_i) \ge u_i(A_j) \implies u_i(A_i) \ge \frac{1}{n}\cdot \sum_{j=1}^{n} u_i(A_j)$.
\begin{corollary}\label{corr:1}
EF $\implies$ PROP-E for arbitrary valuations in presence of full externalities.
\end{corollary}


We now compare PROP-E with existing PROP extensions for capturing externalities. We consider two definitions stated in literature from \cite{externalitiesghodsi} (Average Share) and \cite{aziz2021fairnessexternalities} (General Fair Share). Note that both these definitions are applicable when agents have additive valuations, while PROP-E applies for any general arbitrary valuations. In \cite{aziz2021fairnessexternalities}, the authors proved that Average Share $\implies$ General Fair Share, i.e., if an allocation guarantees all agents their average share value, it also guarantees general fair share value. With that, we state the definition of Average Share (in 2-D) and compare it with PROP-E.
\begin{definition}[Average Share \cite{externalitiesghodsi}]
In $\mathcal{V}$,
the average value of item $k$ for agent $i$, denoted by
\begin{equation}
avg[v_{ik}] = \frac{1}{n} \cdot[v_{ik} + (n-1) v'_{ik}] 
\end{equation}
The average share of agent $i$, 
$\overline{v_i(M)}  = \sum_{k \in M} avg[v_{ik}]$.
An allocation $A$ is said to ensure average share if
$\forall i, u_i(A_i) \ge \overline{v_i(M)}$.
\end{definition}


\begin{proposition}
\label{prop:propeandaverageshare}
PROP-E is equivalent to Average Share in 2-D, for additive valuations. \end{proposition}
\begin{proof} $\forall i \in N,$ 
\begin{equation*}
\begin{aligned}
u_i(A_i) &\geq \frac{1}{n}\cdot\sum_{j\in N} u_i(A_j) 
= \frac{1}{n}\cdot\sum_{j\in N} v_i(A_j) - v'_i(M \setminus A_j) \\
& = \frac{1}{n}\cdot\sum_{k\in M} v_{ik} -\frac{1}{n}\cdot\sum_{j\in N} v'_i(M \setminus A_j) \\
& = \frac{1}{n}\cdot\sum_{k\in M} v_{ik} -\frac{1}{n}\cdot\sum_{k\in M} (n-1) v'_{ik} 
\end{aligned}
\end{equation*}
We can prove the reverse implication in a similar way.
\end{proof}

Next, we briefly state the relation of EF, PROP-E, and Average Share beyond 2-D and give the proofs in the Appendix Section ~\ref{sec:propeaverageshare}.
\begin{remark}
In case of full externality, EF $\centernot\implies$ Average Share \cite{aziz2021fairnessexternalities}. 
\end{remark}

\begin{proposition}
Beyond 2-D, PROP-E $\centernot\implies$ Average Share and Average Share $\centernot\implies$ PROP-E. 
\end{proposition}
To conclude this section, we state that for the special case of 2-D externalities with additive valuations, we can adapt existing algorithms to 2-D, and further analysis is required for the general setting.

Apart from the additive relaxations, the most commonly considered relaxation to PROP is maximin share (MMS) allocations. We provide analysis of MMS for 2-D valuations in the next section.

\section{Approximate MMS in 2-D}
\label{sec:alphamms}
From Theorem~\ref{thm:transformation}, we showed that transformation $\mathfrak{T}$ retains MMS property, i.e., an allocation $A$ guarantees MMS in 1-D \emph{iff} $A$ guarantees MMS in 2-D. We draw attention to the point that, 
\begin{equation}
        \label{eq:mu_itransform}
\mu_i  = \mu^{\mathcal{W}}_i + v'_i(M)
\end{equation}
where $\mu^{\mathcal{W}}_i$ and $\mu_i$ are the MMS value of agent $i$ in 1-D and 2-D, respectively. \cite{procacciammsdoesntexistexample} proved that MMS allocation may not exist even for additive valuations, but multiplicative approximation of MMS always exists in 1-D. 
The current best approximation results on MMS allocation are $3/4+1/(12n)$-MMS for goods \cite{GargMMS2020} and $11/9$-MMS for chores \cite{Xinhuangmmschores} for additive valuations. We are interested in finding multiplicative approximation to MMS in 2-D. Note that we only study $\alpha$-MMS for complete goods or chores in 2-D, as paper \cite{kulkarni2021mixedmanna} proves the non-existence of $\alpha$-MMS in the case of combination of goods and chores in 1-D. From Eq.~\ref{def:mms} of  $\alpha$-MMS, if $\mu_i$ is positive, we consider $\alpha$-MMS allocation with $\alpha \in [0,1]$, and if it is negative, $1/\alpha$-MMS with $\alpha \in (0,1]$.


We categorize externalities in two ways for better analysis 1) Correlated Externality 2) Inverse Externality. In the correlated setting, we study goods with positive externality and chores with negative externality. And in the inverse externality, we study goods with negative externality and chores with positive externality. In the following sections, we show that $\alpha$-MMS exists for correlated, while it may not exist for inverse externality. 
\subsection{$\alpha$-MMS for Correlated Externality}
In this section, we investigate approximate MMS guarantees for correlated externality.

\begin{proposition} For correlated externality, if an allocation $A$ is $\alpha$-MMS in $\mathcal{W}$, $A$ is $\alpha$-MMS in $\mathcal{V}$, but need not vice versa.
\label{prop:analysisforalphamms}
\end{proposition}
\begin{proof}
In the first part of this proof, we prove that $A$ is $\alpha$-MMS in $\mathcal{W}$, $A$ is $\alpha$-MMS in $\mathcal{V}$, and then in the second part, we provide a counter-example such that $A$ is $\alpha$-MMS in $\mathcal{V}$ but not in $\mathcal{W}$.

\noindent \textbf{Part-1.} Let $A$ be $\alpha$-MMS in $\mathcal{W}$,
\begin{equation*}
\begin{aligned}
\forall i \in N, 
w_i(A_i) &\ge \alpha \mu_i^\mathcal{W} & \textit{for goods} \\
u_i(A_i) - v'_i(M)  &\ge  -\alpha v'_i(M) + \alpha \mu_i \\
\forall i \in N, 
w_i(A_i) &\ge \frac{1}{\alpha} \mu_i^\mathcal{W} & \textit{for chores} \\
u_i(A_i) - v'_i(M)  &\ge  -\frac{1}{\alpha} v'_i(M) + \frac{1}{\alpha}  \mu_i 
\end{aligned}
\end{equation*}

In the case of goods with positive externalities, $\mu_i$ is positive, $\alpha \in (0,1]$, and $\forall S \subseteq M$, $v'(S) \ge 0$. From this, we derive that $v_i'(M) \le \alpha v_i'(M)$, and hence it is valid to say that $u_i(A_i) \geq \alpha \mu_i$. In the case of chores with negative externalities, $\mu_ i$ is negative, $1/\alpha \geq 1$, and $\forall S \subseteq M, v'(S) \le 0$.  Similarly to the previous point, we derive that $v'(M) \le \frac{1}{\alpha} v'(M)$ and thus $u_i(A_i) \geq \frac{1}{\alpha}\mu_i$.

\noindent \textbf{Part-2.} 
We provide the following counter-example for goods to prove $A$ is $\alpha$-MMS in $\mathcal{V}$ but not in $\mathcal{W}$. 

\smallskip
\noindent \emph{Example.} Consider $N=\{1,2\}$ both having additive identical valuations for $5$ goods $\{g_1, g_2, g_3, g_4, g_5, g_6\}$ given by, $V_{ig_1}=(0.5,0.1)$, $V_{ig_2}=(0.5,0.1)$, $V_{ig_3}=(0.3,0.1)$,  $V_{ig_4}=(0.5,0.1)$, 
$V_{ig_5}=(0.5,0.1)$, and $V_{ig_6}=(0.5,0.1)$. After transformation, we get $\mu_i^{\mathcal{W}}=1$ and in 2-D $\mu_i=1.6$. Allocation, $A=\{\{g_1\},$ $\{g_2,g_3,g_4,g_5,g_6\}\}$ is $1/2$-MMS in $\mathcal{V}$, but not in $\mathcal{W}$.
We provide the following counter-example for chores with negative externality to prove $A$ is $1/\alpha$-MMS in $\mathcal{V}$ but not in $\mathcal{W}$. 
\begin{example}
\label{example:choreswithneg}
Consider $N=\{1,2\}$ both having additive identical valuations for 3 chores $M=\{c_1,c_2,c_3\}$ given by $V_{1c_1}=(-40,-36)$, $V_{1c_2}=(-110,-70)$, and $V_{1c_3}=(-109,-71)$. Note that $v'_i(M)=-177$. After transformation, we get $\mu_i^{\mathcal{W}}=-42$ and $\mu_i^{\mathcal{V}}=-219$. Let us consider $1/\alpha=4/3$, then $w_i(A_i)\ge -56$ and $u_i(A_i)\ge -292$. Allocation $A=\{(c_1,c_2,c_3),\emptyset\}$ is 4/3-MMS in $\mathcal{V}$, but not in $\mathcal{W}$. 
\end{example}
\end{proof}

\begin{corollary}
 We can adapt the existing $\alpha$-MMS algorithms using $\mathfrak{T}$ for correlated externality for \emph{general valuations}.
\end{corollary}
\begin{corollary}
For correlated 2-D externality, we can always obtain $3/4+1/(12n)$-MMS for goods and $11/9$-MMS for chores for \emph{additive}.
\end{corollary}

\subsection{$\alpha$-MMS for Inverse Externality}
Motivated by the example given in \cite{procacciammsdoesntexistexample} for non-existence of MMS allocation for 1-D valuations, we ingeniously adapted it to construct the following instance in 2-D to prove the impossibility of $\alpha$-MMS in 2-D. 
We show that for any $\alpha \in (0,1]$, an $\alpha$-MMS or $1/\alpha$-MMS allocation may not exist for inverse externality. In this section, we present an instance for goods with inverse externality, such that $\forall i, \mu_i$ is positive, we show that for $\alpha >0$, there is no $\alpha$-MMS allocation. Further, we present an instance for chores with positive externality, such that $\forall i, \mu_i$ is negative, and we prove that there is no $1/\alpha$ MMS allocation. In order to prove the non-existence results for goods we consider the 1-D valuations $W$ where MMS does not exist. Further we use this example to construct $V^g$ where $W = \mathfrak{T}(V^g)$, in 2-D such that $\alpha$-MMS exists in $V^{g}$ only if MMS allocation exists in $W$. Hence the contradiction.

\smallskip
\noindent\emph{\underline{Non-existence of $\alpha$-MMS in Goods.}}
Consider the following example.
\begin{example}
\label{example:goodswithpos}
\emph{We consider a problem of allocating 12 goods among three agents, and represent valuation profile as $V^{g}$. The valuation profile $V^{g}$ is equivalent to $10^3 \times V$ given in Table~\ref{tab:additive2Dvaluationprofile}. We set $\epsilon_1 = 10^{-4}$ and $\epsilon_2=10^{-3}$. We transform these valuations in 1-D using $\mathfrak{T}$, and the valuation profile $\mathfrak{T}(V^g)$ is the same as the instance in \cite{procacciammsdoesntexistexample} that proves the non-existence of MMS for goods. Note that $\forall i, v_i'(M)=-4055000+10^3\epsilon_1$ 
The MMS value of every agent in $\mathfrak{T}(V^{g})$ is 4055000 and from Eq.~\ref{eq:mu_itransform}, the MMS value of every agent in $V^{g}$ is $10^3\epsilon_1$.}
\end{example}

Recall that $\mathfrak{T}$ retains MMS property (Theorem ~\ref{thm:transformation}) and thus we can say that MMS allocation doesn't exist in $V^g$.
\begin{table}[!t]
\caption{Additive 2-D Valuation Profile ($V$)}
\begin{center}
\begin{tabular}{llll} 
\toprule

\multicolumn{1}{c}{\multirow{2}{*}{Item}}
& \multicolumn{1}{c}{Agent $1$} & \multicolumn{1}{c}{Agent $2$}  & \multicolumn{1}{c}{Agent $3$}  \\ 
 & \multicolumn{1}{c}{$(v_1,v'_1)$} & \multicolumn{1}{c}{$(v_2,v'_2)$} & \multicolumn{1}{c}{$(v_3,v'_3)$} \\ 
\midrule
 \arrayrulecolor{lightgray}
\multirow{2}{*}{$k_1$}& 
\multicolumn{1}{c}{\multirow{2}{*}{\begin{tabular}[c]{@{}c@{}}(3$\epsilon_2$, -1017+3$\epsilon_1$-3$\epsilon_2$)\end{tabular}}}
&
\multicolumn{1}{c}{\multirow{2}{*}{\begin{tabular}[c]{@{}c@{}}(3$\epsilon_2$, -1017+3$\epsilon_1$-3$\epsilon_2$)\end{tabular}}}
& \multicolumn{1}{c}{\multirow{2}{*}{\begin{tabular}[c]{@{}c@{}}(3$\epsilon_2$, -1017+3$\epsilon_1$-3$\epsilon_2$)\end{tabular}}} \\ \\
\hline

\multirow{2}{*}{$k_2$}&\multicolumn{1}{c}{\multirow{2}{*}{\begin{tabular}[c]{@{}c@{}}$(2\epsilon_1$, -1025+$2\epsilon_1$+$\epsilon_2$)\end{tabular}}}
&\multicolumn{1}{c}{\multirow{2}{*}{\begin{tabular}[c]{@{}c@{}}$(2\epsilon_1$, -1025+$2\epsilon_1$+$\epsilon_2$)\end{tabular}}}
& \multicolumn{1}{c}{\multirow{2}{*}{\begin{tabular}[c]{@{}c@{}}$(1025-\epsilon_1$, -$\epsilon_1)$\end{tabular}}}\\ \\
\hline
\multirow{2}{*}{$k_3$}&\multicolumn{1}{c}{\multirow{2}{*}{\begin{tabular}[c]{@{}c@{}}$(2\epsilon_1$, -1012+$2\epsilon_1$+$\epsilon_2$)\end{tabular}}}
&\multicolumn{1}{c}{\multirow{2}{*}{\begin{tabular}[c]{@{}c@{}}($1012-\epsilon_1$, -$\epsilon_1$)\end{tabular}}}
& \multicolumn{1}{c}{\multirow{2}{*}{\begin{tabular}[c]{@{}c@{}} ($2\epsilon_1$, -1012+$2\epsilon_1$+$\epsilon_2$)\end{tabular}}}\\ \\
\hline
\multirow{2}{*}{$k_4$}&\multicolumn{1}{c}{\multirow{2}{*}{\begin{tabular}[c]{@{}c@{}}($2\epsilon_1$, -1001+$2\epsilon_1$+$\epsilon_2$)\end{tabular}}}
&\multicolumn{1}{c}{\multirow{2}{*}{\begin{tabular}[c]{@{}c@{}}($1001-\epsilon_1$, -$\epsilon_1$)\end{tabular}}}
& \multicolumn{1}{c}{\multirow{2}{*}{\begin{tabular}[c]{@{}c@{}}($1001-\epsilon_1$, -$\epsilon_1$)\end{tabular}}}\\ \\
\hline
\multirow{2}{*}{$k_5$}&\multicolumn{1}{c}{\multirow{2}{*}{\begin{tabular}[c]{@{}c@{}}$(1002-\epsilon_1$, -$\epsilon_1$)\end{tabular}}}
&\multicolumn{1}{c}{\multirow{2}{*}{\begin{tabular}[c]{@{}c@{}}($2\epsilon_1$, -1002+$2\epsilon_1$+$\epsilon_2$)\end{tabular}}}
& \multicolumn{1}{c}{\multirow{2}{*}{\begin{tabular}[c]{@{}c@{}} ($1002-\epsilon_1$, -$\epsilon_1$)\end{tabular}}}\\ \\
\hline
\multirow{2}{*}{$k_6$}&\multicolumn{1}{c}{\multirow{2}{*}{\begin{tabular}[c]{@{}c@{}}($1022-\epsilon_1$, -$\epsilon_1$)\end{tabular}}}
&\multicolumn{1}{c}{\multirow{2}{*}{\begin{tabular}[c]{@{}c@{}}($1022-\epsilon_1$, -$\epsilon_1$)\end{tabular}}}
& \multicolumn{1}{c}{\multirow{2}{*}{\begin{tabular}[c]{@{}c@{}} ($1022-\epsilon_1$, -$\epsilon_1$)\end{tabular}}}\\ \\
\hline
\multirow{2}{*}{$k_7$}&\multicolumn{1}{c}{\multirow{2}{*}{\begin{tabular}[c]{@{}c@{}}($1003-\epsilon_1$, -$\epsilon_1$)\end{tabular}}}
&\multicolumn{1}{c}{\multirow{2}{*}{\begin{tabular}[c]{@{}c@{}}($1003-\epsilon_1$, -$\epsilon_1$)\end{tabular}}}
& \multicolumn{1}{c}{\multirow{2}{*}{\begin{tabular}[c]{@{}c@{}} ($2\epsilon_1$, -1003+$2\epsilon_1$+$\epsilon_2$)\end{tabular}}}\\ \\
\hline
\multirow{2}{*}{$k_8$}&\multicolumn{1}{c}{\multirow{2}{*}{\begin{tabular}[c]{@{}c@{}}($1028-\epsilon_1$, -$\epsilon_1$)\end{tabular}}}
&\multicolumn{1}{c}{\multirow{2}{*}{\begin{tabular}[c]{@{}c@{}}($1028-\epsilon_1$, -$\epsilon_1$)\end{tabular}}}
& \multicolumn{1}{c}{\multirow{2}{*}{\begin{tabular}[c]{@{}c@{}} ($1028-\epsilon_1$, -$\epsilon_1$)\end{tabular}}}\\ \\
\hline
\multirow{2}{*}{$k_9$}&\multicolumn{1}{c}{\multirow{2}{*}{\begin{tabular}[c]{@{}c@{}}($1011-\epsilon_1$, -$\epsilon_1$)\end{tabular}}}
&\multicolumn{1}{c}{\multirow{2}{*}{\begin{tabular}[c]{@{}c@{}}($2\epsilon_1$, -1011+$2\epsilon_1$+$\epsilon_2$)\end{tabular}}}
& \multicolumn{1}{c}{\multirow{2}{*}{\begin{tabular}[c]{@{}c@{}} ($1011-\epsilon_1$, -$\epsilon_1$)\end{tabular}}}\\ \\
\hline
\multirow{2}{*}{$k_{10}$}&\multicolumn{1}{c}{\multirow{2}{*}{\begin{tabular}[c]{@{}c@{}}($1000-\epsilon_1$, -$\epsilon_1$)\end{tabular}}}
&\multicolumn{1}{c}{\multirow{2}{*}{\begin{tabular}[c]{@{}c@{}}($1000-\epsilon_1$, -$\epsilon_1$)\end{tabular}}}
& \multicolumn{1}{c}{\multirow{2}{*}{\begin{tabular}[c]{@{}c@{}} ($1000-\epsilon_1$, -$\epsilon_1$)\end{tabular}}}\\ \\
\hline
\multirow{2}{*}{$k_{11}$}&\multicolumn{1}{c}{\multirow{2}{*}{\begin{tabular}[c]{@{}c@{}}($1021-\epsilon_1$, -$\epsilon_1$)\end{tabular}}}
&\multicolumn{1}{c}{\multirow{2}{*}{\begin{tabular}[c]{@{}c@{}}($1021-\epsilon_1$, -$\epsilon_1$)\end{tabular}}}
& \multicolumn{1}{c}{\multirow{2}{*}{\begin{tabular}[c]{@{}c@{}} ($1021-\epsilon_1$, -$\epsilon_1$)\end{tabular}}}\\ \\
\hline
\multirow{2}{*}{$k_{12}$}&\multicolumn{1}{c}{\multirow{2}{*}{\begin{tabular}[c]{@{}c@{}}($1023-\epsilon_1$, -$\epsilon_1$)\end{tabular}}}
&\multicolumn{1}{c}{\multirow{2}{*}{\begin{tabular}[c]{@{}c@{}}($1023-\epsilon_1$, -$\epsilon_1$)\end{tabular}}}
& \multicolumn{1}{c}{\multirow{2}{*}{\begin{tabular}[c]{@{}c@{}}($2\epsilon_1$, -1023+$2\epsilon_1$+$\epsilon_2$)\end{tabular}}}\\ \\

 \arrayrulecolor{black}
\bottomrule
\end{tabular}
\end{center}

\label{tab:additive2Dvaluationprofile}
\end{table}


\begin{lemma}
\label{lemma:goodswithneg}
There is no $\alpha$-MMS allocation for the valuation profile $V^{g}$ of Example~\ref{example:goodswithpos} for any $\alpha\in[0,1]$.
\end{lemma}
\begin{proof}
An allocation $A$ is $\alpha$-MMS for $\alpha \ge 0$ \emph{iff} $\forall i, u_i(A_i)\ge \alpha \mu_i \ge 0$ when $\mu_i>0$. Note that the transformed valuations $w_i(A_i) = \mathfrak{T}(V^g_i(A_i))$. From Eq.~\ref{eq:transformTus}, $u_i(A_i)\ge 0$, \emph{iff} $w_i(A_i) \ge -v'_i(M)$, which gives us $w_i(A_i)\ge4055000-0.1$. For this to be true, we need $w_i(A_i) \ge 4055000$ since $\mathfrak{T}(V^g)$ has all integral values. We know that such an allocation doesn't exist~\cite{procacciammsdoesntexistexample}. Hence for any $\alpha \in [0,1]$, $\alpha$-MMS does not exist for $V^g$.
\end{proof}

\noindent\emph{\underline{Non-existence of $1/\alpha$-MMS in Chores.}}
Consider the following example.
\begin{example}
\label{example:choreswithpos}
\emph{We consider a problem of allocating 12 chores among three agents. The valuation profile $V^{c}$ is equivalent to $-10^3V$ given in Table~\ref{tab:additive2Dvaluationprofile}. We set $\epsilon_2=-10^{-3}$. We transform these valuations in 1-D, and $\mathfrak{T}(V^{c})$ is the same as the instance in \cite{azizmmschores} that proves the non-existence of MMS for chores. Note that $v'_i(M)=4055000-10^3\epsilon_1$.  The MMS value of every agent in $\mathfrak{T}(V^{c})$ amd $V^{c}$ is -4055000 and $-10^3\epsilon_1$, respectively. }
\end{example}


\begin{lemma}
\label{lemma:1alphammsdoesntexsist}
There is no $1/\alpha$-MMS allocation for the valuation profile $V^{c}$ of Example~\ref{example:choreswithpos} with $\epsilon_1 \in (0,10^{-4}]$ for any $\alpha>0$.
\end{lemma}
\begin{proof}
An allocation $A$ is $1/\alpha$-MMS for $\alpha>0$ $\emph{iff}$ $\forall i, u_i(A_i)\ge \frac{1}{\alpha} \mu_i$ when $\mu_i<0$. We set $\epsilon_1 \le 10^{-4}$ in $V^c$. When $\alpha\ge10^3\epsilon_1$ $\forall i$ then $u_i(A_i) \ge -1$.
From Eq.~\ref{eq:transformTus}, $u_i(A_i)\ge -1$ \emph{iff} $w_i(A_i)\ge-4055001+10^3\epsilon_1$. Note that $0<10^3\epsilon_1\le0.1$ and since $w_i(A_i)$ has only integral values, we need $\forall i, w_i(A_i)\ge -4055000$. Such $A$ does not exist \cite{azizmmschores}. As $\epsilon_1$ decreases, $1/\epsilon_1$ increases, and even though approximation guarantees weakens, it still does not exist for $V^{c}$. 
\end{proof}

From Lemma~\ref{lemma:goodswithneg} and \ref{lemma:1alphammsdoesntexsist} we conclude the following theorem,
\begin{theorem}
\label{thm:mulmmsdoesntexist}
There may not exist $\alpha$-MMS for any $\alpha\in[0,1]$ for $\mu_i>0$ or $1/\alpha$-MMS allocation for any $\alpha \in (0,1]$ for $\mu_i<0$ in the presence of externalities.
\end{theorem}

Interestingly, in 1-D, $\alpha$-MMS's non-existence is known for $\alpha$ value close to 1 \cite{procacciammsdoesntexistexample,feige2021tight}, while in 2-D, it need not exist even for $\alpha=0$. It follows because $\alpha$-MMS may not be lead to any relaxation in the presence of inverse externalities.

Consider the situation of goods having negative externalities, where MMS share $\mu_i$ comprises of the positive value from the assigned bundle $A_i$ and negative value from the unassigned bundles $A_{-i}$. 
 We re-write $\mu_i$ as follows, $\mu_i=\mu_i^{+} + \mu_i^{-}$ where $\mu_i^+$ corresponds to utility from assigned goods/unassigned chores and $\mu_i^-$ corresponds to utility from unassigned goods/assigned chores. When $\mu_i \ge 0$, applying $\alpha \mu_i$ is not only relaxing positive value $\alpha \mu_i^+$, but also requires $\alpha \mu_i^-$ which is stricter than $\mu_i^-$ since $\mu_i^- < 0$. Hence, the impossibility of  $\alpha$-MMS in 2-D. Similar argument holds for chores with positive externalities.

In the next section, we explore relaxing MMS such that it is guaranteed to exist in 2-D.



\subsection{Re-defining Approximate MMS}
\label{sec:shiftedalphamms}
In this section, we define \emph{Shifted $\alpha$-MMS} that guarantees a fraction of MMS share shifted by certain value, such that it always exist in 2-D. We also considered intuitive ways of approximating MMS in 2-D. These ways are based on relaxing the positive value obtained from MMS allocation $\mu^+$ and the negative value $\mu^-$, $\mu = \mu^+ + \mu^-$. In other words, we look for allocations that guarantee $\alpha \mu^+$ and $(1+\alpha)$ or $1/\alpha$ of $\mu^-$. Unfortunately, such approximations may not always exist. We provide detailed explanation in the Appendix Section~\ref{sec:alphammsotherdefinition}. 



\begin{definition}[Shifted $\alpha$-MMS]
\label{def:alphamms3}
An allocation $A$ guarantees shifted $\alpha$-MMS 
if $\forall i \in N, \alpha\in(0,1]$
\begin{align}
    & \begin{rcases}
 u_i(A_i) \ge \alpha \mu_i + (1-\alpha) v'_i(M)\\
    \end{rcases} 
    &\mbox{ for goods }  \nonumber \\
     & \begin{rcases}
 u_i(A_i) \ge \frac{1}{\alpha}\mu_i + \frac{\alpha-1}{\alpha} v'_i(M)
    \end{rcases} 
    & \mbox{ for chores} \nonumber
\end{align}
\end{definition}

\begin{proposition}
\label{prop:shiftedalphamms}
An allocation $A$ is $\alpha$-MMS in $\mathcal{W}$ \emph{iff} $A$ is shifted $\alpha$-MMS (Def.~\ref{def:alphamms3}) in $\mathcal{V}$.
\end{proposition}
\begin{proof}
For goods, if allocation $A$ is shifted $\alpha$-MMS, $\forall i , u_i(A_i) \ge \alpha \mu_i + (1-\alpha)v'_i(M)$.
Applying $\mathfrak{T}$, we get $w_i(A_i) + v_i'(M) \ge \alpha\mu_i^{\mathcal{W}} + \alpha v'_i(M)  + (1-\alpha)v'_i(M)$ which gives $w_i(A_i) \ge \alpha \mu_i^{\mathcal{W}}$.
For chores, if $A$ is shifted $1/\alpha$-MMS, $\forall i, u_i(A_i) \ge \frac{1}{\alpha}\mu_i + \frac{(\alpha-1)}{\alpha}v'_i(M)$. Applying $\mathfrak{T}$ gives $w_i(A_i) \ge \frac{1}{\alpha}\mu_i^{\mathcal{W}}$.
Similarly we can prove vice versa.
\end{proof}

\begin{corollary} We can adapt all the existing algorithms for $\alpha$-MMS in $\mathcal{W}$ to get shifted $\alpha$-MMS in $\mathcal{V}$.
\end{corollary}

We use $\mathfrak{T}$ and apply the existing algorithms for additive valuations as well as general valuations and obtain the corresponding shifted multiplicative approximations.
Since a direct multiplicative approximation of MMS need not exist in presence of externalities, we consider additive relaxation of MMS in the next section.

\subsubsection{Additive Relaxation of MMS}
\begin{definition}[MMS relaxations]
\label{def:mms1}
An allocation $A$ that satisfies, $\forall i,j \in N$, MMSX i.e., MMS upto any item,
\begin{align}
    & \begin{rcases}
\forall \ k \in \{M \setminus A_i\}, 
v_{ik}>0,      u_i(A_i \cup \{k\}) \ge \mu_i\\
\forall \ k \in A_i,v_{ik}<0, u_i(A_i \setminus \{k\}) \ge   \mu_i
    \end{rcases} \\
    &\mbox{ satisfies MMS1 (Maximin Share up to an item)} \nonumber \\
     & \begin{rcases}
\exists \ k \in \{M \setminus A_i\},     u_i(A_i \cup \{k\}) \ge \mu_i,  \mbox{ or,}\\
\exists \ k \in A_i, u_i(A_i \setminus \{k\}) \ge \mu_i
    \end{rcases} 
\end{align}
\end{definition}

\begin{proposition}
From Theorem~\ref{thm:transformation}, we conclude that MMS1 and MMSX retains after transformation.
\end{proposition}
EF1 is a stronger fairness notion than MMS1 and can be computed in polynomial time. On the other hand, PROPX might not exist for goods~\cite{PROPXAziz2019}. Since PROPX implies MMSX, it is interesting to settle the existence of MMSX for goods. Note that MMSX and Shifted $\alpha$-MMS are not related. It is interesting to study these relaxations further, even in full externalities.




\section{Conclusion}
In this paper, we conducted a study on indivisible item allocation with special externalities -- 2-D externalities. We proposed a simple yet compelling transformation from 2-D to 1-D to employ existing algorithms to ensure many fairness and efficiency notions. We can adapt existing fair division algorithms via the transformation in such settings. We proposed proportionality extension in the presence of externalities and studied its relation with other fairness notions. For MMS fairness, we proved the impossibility of multiplicative approximation of MMS in 2-D, and we proposed Shifted $\alpha$-MMS instead. There are many exciting questions here which we leave for future works. (i) It might be impossible to have fairness-preserving valuation transformation for general externalities. However, what are some interesting domains where such transformations exist? (ii) What are interesting approximations to MMS in 2-D as well as in general externalities?




\bibliographystyle{unsrtnat}
\bibliography{references}  

\begin{thebibliography}{40}
\providecommand{\natexlab}[1]{#1}
\providecommand{\url}[1]{\texttt{#1}}
\expandafter\ifx\csname urlstyle\endcsname\relax
  \providecommand{\doi}[1]{doi: #1}\else
  \providecommand{\doi}{doi: \begingroup \urlstyle{rm}\Url}\fi

\bibitem[Brams et~al.(1996)Brams, Brams, and Taylor]{brams_taylor_1996}
Steven~J Brams, Steven~John Brams, and Alan~D Taylor.
\newblock \emph{Fair Division: From cake-cutting to dispute resolution}.
\newblock Cambridge University Press, 1996.

\bibitem[Moulin(2004)]{moulinfairness}
Herv{\'e} Moulin.
\newblock \emph{Fair division and collective welfare}.
\newblock MIT press, 2004.

\bibitem[Segal-Halevi(2019)]{Segal_Halevi_2019}
Erel Segal-Halevi.
\newblock Cake-cutting with different entitlements: How many cuts are needed?
\newblock \emph{Journal of Mathematical Analysis and Applications},
  480\penalty0 (1):\penalty0 123382, 2019.

\bibitem[Steihaus(1948)]{propSTEIHAUS}
Hugo Steihaus.
\newblock The problem of fair division.
\newblock \emph{Econometrica}, 16:\penalty0 101--104, 1948.

\bibitem[Su(2000)]{webbcakecutting}
Francis~Edward Su.
\newblock Cake-cutting algorithms: Be fair if you can. by jack robertson and
  william webb.
\newblock \emph{The American Mathematical Monthly}, 107\penalty0 (2):\penalty0
  185--188, 2000.

\bibitem[Caragiannis et~al.(2019)Caragiannis, Kurokawa, Moulin, Procaccia,
  Shah, and Wang]{Caragiannis2016}
Ioannis Caragiannis, David Kurokawa, Herv{\'e} Moulin, Ariel~D Procaccia,
  Nisarg Shah, and Junxing Wang.
\newblock The unreasonable fairness of maximum nash welfare.
\newblock \emph{ACM Transactions on Economics and Computation (TEAC)},
  7\penalty0 (3):\penalty0 1--32, 2019.

\bibitem[Barman et~al.(2018{\natexlab{a}})Barman, Krishnamurthy, and
  Vaish]{barman2018finding}
Siddharth Barman, Sanath~Kumar Krishnamurthy, and Rohit Vaish.
\newblock Finding fair and efficient allocations.
\newblock In \emph{Proceedings of the 2018 ACM Conference on Economics and
  Computation}, pages 557--574, 2018{\natexlab{a}}.

\bibitem[De~Keijzer et~al.(2009)De~Keijzer, Bouveret, Klos, and
  Zhang]{computationallyguarantees1}
Bart De~Keijzer, Sylvain Bouveret, Tomas Klos, and Yingqian Zhang.
\newblock On the complexity of efficiency and envy-freeness in fair division of
  indivisible goods with additive preferences.
\newblock In \emph{International Conference on Algorithmic Decision Theory},
  pages 98--110. Springer, 2009.

\bibitem[Freeman et~al.(2019)Freeman, Sikdar, Vaish, and Xia]{EQ1freemangoods}
Rupert Freeman, Sujoy Sikdar, Rohit Vaish, and Lirong Xia.
\newblock Equitable allocations of indivisible goods.
\newblock In \emph{IJCAI}, 2019.

\bibitem[Procaccia and Wang(2014)]{procaccia2014mms}
Ariel~D Procaccia and Junxing Wang.
\newblock Fair enough: Guaranteeing approximate maximin shares.
\newblock In \emph{Proceedings of the fifteenth ACM conference on Economics and
  computation}, pages 675--692, 2014.

\bibitem[Foley(1966)]{Foley1967ResourceAA}
Duncan~Karl Foley.
\newblock \emph{Resource allocation and the public sector}.
\newblock Yale University, 1966.

\bibitem[Garg and Taki(2021)]{GargMMS2020}
Jugal Garg and Setareh Taki.
\newblock An improved approximation algorithm for maximin shares.
\newblock \emph{Artificial Intelligence}, page 103547, 2021.

\bibitem[Huang and Lu(2021)]{Xinhuangmmschores}
Xin Huang and Pinyan Lu.
\newblock An algorithmic framework for approximating maximin share allocation
  of chores.
\newblock In \emph{Proceedings of the 22nd ACM Conference on Economics and
  Computation}, pages 630--631, 2021.

\bibitem[Velez(2016)]{velez2016fairness}
Rodrigo~A Velez.
\newblock Fairness and externalities.
\newblock \emph{Theoretical Economics}, 11\penalty0 (1):\penalty0 381--410,
  2016.

\bibitem[Br{\^a}nzei et~al.(2013)Br{\^a}nzei, Procaccia, and
  Zhang]{externalitiesdivisbleijcai13}
Simina Br{\^a}nzei, Ariel Procaccia, and Jie Zhang.
\newblock Externalities in cake cutting.
\newblock In \emph{Twenty-Third International Joint Conference on Artificial
  Intelligence}, 2013.

\bibitem[Treibich(2019)]{treibich2019welfare}
Rafael Treibich.
\newblock Welfare egalitarianism with other-regarding preferences.
\newblock \emph{Social Choice and Welfare}, 52\penalty0 (1):\penalty0 1--28,
  2019.

\bibitem[Seddighin et~al.(2019)Seddighin, Saleh, and
  Ghodsi]{externalitiesghodsi}
Masoud Seddighin, Hamed Saleh, and Mohammad Ghodsi.
\newblock Externalities and fairness.
\newblock In \emph{The World Wide Web Conference}, pages 538--548, 2019.

\bibitem[Aziz et~al.(2021)Aziz, Suksompong, Sun, and
  Walsh]{aziz2021fairnessexternalities}
Haris Aziz, Warut Suksompong, Zhaohong Sun, and Toby Walsh.
\newblock Fairness concepts for indivisible items with externalities, 2021.

\bibitem[Plaut and Roughgarden(2020)]{EFXsolution}
Benjamin Plaut and Tim Roughgarden.
\newblock Almost envy-freeness with general valuations.
\newblock \emph{SIAM Journal on Discrete Mathematics}, 34\penalty0
  (2):\penalty0 1039--1068, 2020.

\bibitem[Aziz et~al.(2020{\natexlab{a}})Aziz, Moulin, and
  Sandomirskiy]{PROPXAziz2019}
Haris Aziz, Herv{\'e} Moulin, and Fedor Sandomirskiy.
\newblock A polynomial-time algorithm for computing a pareto optimal and almost
  proportional allocation.
\newblock \emph{Operations Research Letters}, 48\penalty0 (5):\penalty0
  573--578, 2020{\natexlab{a}}.

\bibitem[Chaudhury et~al.(2020)Chaudhury, Garg, and Mehlhorn]{chaudhury2020efx}
Bhaskar~Ray Chaudhury, Jugal Garg, and Kurt Mehlhorn.
\newblock Efx exists for three agents.
\newblock In \emph{Proceedings of the 21st ACM Conference on Economics and
  Computation}, pages 1--19, 2020.

\bibitem[Budish(2011)]{Budish11}
Eric Budish.
\newblock The combinatorial assignment problem: Approximate competitive
  equilibrium from equal incomes.
\newblock \emph{Journal of Political Economy}, 119\penalty0 (6):\penalty0
  1061--1103, 2011.

\bibitem[Lipton et~al.(2004)Lipton, Markakis, Mossel, and Saberi]{Lipton2004}
Richard~J Lipton, Evangelos Markakis, Elchanan Mossel, and Amin Saberi.
\newblock On approximately fair allocations of indivisible goods.
\newblock In \emph{Proceedings of the 5th ACM Conference on Electronic
  Commerce}, pages 125--131, 2004.

\bibitem[Bhaskar et~al.(2020)Bhaskar, Sricharan, and Vaish]{vaishenvycyle}
Umang Bhaskar, A.~R. Sricharan, and Rohit Vaish.
\newblock On approximate envy-freeness for indivisible chores and mixed
  resources.
\newblock \emph{CoRR}, abs/2012.06788, 2020.
\newblock URL \url{https://arxiv.org/abs/2012.06788}.

\bibitem[Aziz et~al.(2018)Aziz, Caragiannis, and Igarashi]{azizdoublerrijcai19}
Haris Aziz, Ioannis Caragiannis, and Ayumi Igarashi.
\newblock Fair allocation of combinations of indivisible goods and chores.
\newblock \emph{CoRR}, abs/1807.10684, 2018.
\newblock URL \url{http://arxiv.org/abs/1807.10684}.

\bibitem[Barman et~al.(2019)Barman, Ghalme, Jain, Kulkarni, and
  Narang]{barmanaamas19strategicagents}
Siddharth Barman, Ganesh Ghalme, Shweta Jain, Pooja Kulkarni, and Shivika
  Narang.
\newblock Fair division of indivisible goods among strategic agents.
\newblock In \emph{Proceedings of the 18th International Conference on
  Autonomous Agents and MultiAgent Systems}, AAMAS '19, page 1811–1813,
  Richland, SC, 2019. International Foundation for Autonomous Agents and
  Multiagent Systems.
\newblock ISBN 9781450363099.

\bibitem[Bei et~al.(2020)Bei, Huzhang, and Suksompong]{bei2020truthful}
Xiaohui Bei, Guangda Huzhang, and Warut Suksompong.
\newblock Truthful fair division without free disposal.
\newblock \emph{Social Choice and Welfare}, 55\penalty0 (3):\penalty0 523--545,
  2020.

\bibitem[Padala and Gujar(2021)]{padala2021mechanism}
Manisha Padala and Sujit Gujar.
\newblock Mechanism design without money for fair allocations, 2021.

\bibitem[Li et~al.(2021)Li, Li, and Wu]{li2021almost}
Bo~Li, Yingkai Li, and Xiaowei Wu.
\newblock Almost (weighted) proportional allocations for indivisible chores.
\newblock \emph{arXiv preprint arXiv:2103.11849}, 2021.

\bibitem[Kurokawa et~al.(2016)Kurokawa, Procaccia, and Wang]{kurokawa2016}
David Kurokawa, Ariel~D Procaccia, and Junxing Wang.
\newblock When can the maximin share guarantee be guaranteed?
\newblock In \emph{Thirtieth AAAI Conference on Artificial Intelligence}, 2016.

\bibitem[Amanatidis et~al.(2017)Amanatidis, Markakis, Nikzad, and
  Saberi]{Amanatidismms2017}
Georgios Amanatidis, Evangelos Markakis, Afshin Nikzad, and Amin Saberi.
\newblock Approximation algorithms for computing maximin share allocations.
\newblock \emph{ACM Transactions on Algorithms (TALG)}, 13\penalty0
  (4):\penalty0 1--28, 2017.

\bibitem[Barman et~al.(2018{\natexlab{b}})Barman, Biswas, Krishnamurthy, and
  Narahari]{Barman2018GroupwiseMF}
Siddharth Barman, Arpita Biswas, Sanath Krishnamurthy, and Yadati Narahari.
\newblock Groupwise maximin fair allocation of indivisible goods.
\newblock In \emph{Proceedings of the AAAI Conference on Artificial
  Intelligence}, volume~32, 2018{\natexlab{b}}.

\bibitem[Garg et~al.(2019)Garg, McGlaughlin, and Taki]{garg2018}
Jugal Garg, Peter McGlaughlin, and Setareh Taki.
\newblock Approximating maximin share allocations.
\newblock \emph{Open access series in informatics}, 69, 2019.

\bibitem[Ghodsi et~al.(2018)Ghodsi, HajiAghayi, Seddighin, Seddighin, and
  Yami]{ghodsi2017mms}
Mohammad Ghodsi, MohammadTaghi HajiAghayi, Masoud Seddighin, Saeed Seddighin,
  and Hadi Yami.
\newblock Fair allocation of indivisible goods: Improvements and
  generalizations.
\newblock In \emph{Proceedings of the 2018 ACM Conference on Economics and
  Computation}, pages 539--556, 2018.

\bibitem[Aziz et~al.(2017)Aziz, Rauchecker, Schryen, and Walsh]{azizmmschores}
Haris Aziz, Gerhard Rauchecker, Guido Schryen, and Toby Walsh.
\newblock Algorithms for max-min share fair allocation of indivisible chores.
\newblock AAAI'17, page 335–341. AAAI Press, 2017.

\bibitem[Kulkarni et~al.(2021)Kulkarni, Mehta, and
  Taki]{kulkarni2021mixedmanna}
Rucha Kulkarni, Ruta Mehta, and Setareh Taki.
\newblock Indivisible mixed manna: On the computability of mms+ po allocations.
\newblock In \emph{Proceedings of the 22nd ACM Conference on Economics and
  Computation}, pages 683--684, 2021.

\bibitem[Aziz et~al.(2020{\natexlab{b}})Aziz, Huang, Mattei, and
  Segal{-}Halevi]{azizef1um}
Haris Aziz, Xin Huang, Nicholas Mattei, and Erel Segal{-}Halevi.
\newblock Computing fair utilitarian allocations of indivisible goods.
\newblock \emph{CoRR}, abs/2012.03979, 2020{\natexlab{b}}.
\newblock URL \url{https://arxiv.org/abs/2012.03979}.

\bibitem[Chen and Liu(2020)]{chen2020fairness}
Xingyu Chen and Zijie Liu.
\newblock The fairness of leximin in allocation of indivisible chores.
\newblock \emph{arXiv preprint arXiv:2005.04864}, 2020.

\bibitem[Kurokawa et~al.(2018)Kurokawa, Procaccia, and
  Wang]{procacciammsdoesntexistexample}
David Kurokawa, Ariel~D Procaccia, and Junxing Wang.
\newblock Fair enough: Guaranteeing approximate maximin shares.
\newblock \emph{Journal of the ACM (JACM)}, 65\penalty0 (2):\penalty0 1--27,
  2018.

\bibitem[Feige et~al.(2021)Feige, Sapir, and Tauber]{feige2021tight}
Uriel Feige, Ariel Sapir, and Laliv Tauber.
\newblock A tight negative example for mms fair allocations.
\newblock \emph{arXiv preprint arXiv:2104.04977}, 2021.

\end{thebibliography}

\section*{Appendix}
\appendix
\section{Complete Proof of Theorem 1}
\label{sec:completeproofoftheorem1}

\begin{theorem}
\label{thm:transformation}
An Allocation $A$ is $\mathfrak{F}$-Fair and $\mathfrak{E}$-Efficient in $\mathcal{V}$ \emph{iff} $A$ is $\mathfrak{F}$-Fair and $\mathfrak{E}$-Efficient in the transformed 1-D $\mathcal{W}$, where $\mathfrak{F} \in \{\mbox{EF, EF1, EFX, PROP-E, PROP1-E, PROP1-X, MMS}\}$ and $\mathfrak{E} \in \{\mbox{PO, MUW}\}$.
\end{theorem}
\begin{proof}
\noindent\textbf{Fairness notions.}


Considering $\mathfrak{F} =$ EF1. An allocation is EF1 if $\forall i,j \in N, u_i(A_i) \ge u_i(A_j)$, or $\exists k \in \{A_i \cup A_j\}$, s.t., $u_i(A_i \setminus \{ k\}) \ge u_i(A_j \setminus \{k\})$. When $k$ is good for agent $i$, then $k \in A_j$ and when $k$ is chore to agent $i$, then $k \in A_i$.
Let allocation $A$ be EF1 in $\mathcal{W}$ then, $\forall i, \forall j \in N, \exists k \in A_j$, i.e., $k$ is a good for agent $i$ in $A_j$ bundle
\begin{equation*}
\begin{aligned}
w_i(A_i) &\ge w_i(A_j \setminus \{k\}) \\
v_i(A_i) + v'_i(A_{-i})
&\ge v_i(A_j \setminus \{k\}) + v'_i(A_{-j} \cup \{k\})  \\
u_i(A_i)  &\ge u_j(A_j\setminus \{k\})  
\end{aligned}
\end{equation*}

When $k \in A_i$, i.e., $k$ is a chore for agent $i$,
\begin{equation*}
\begin{aligned}
w_i(A_i \setminus \{k\}) &\ge w_i(A_j) \\
v_i(A_i \setminus \{k\}) + v'_i(A_{-i} \cup \{k\}) 
&\ge v_i(A_j ) + v'_i(A_{-j} )  \\
u_i(A_i \setminus \{k\})  &\ge u_j(A_j)  
\end{aligned}
\end{equation*}
The reverse implication follows similarly.
The proof of EFX is similar to that of EF1.

Moving on to PROP-E. Consider $\mathfrak{F}=PROP-E$. An allocation is said to be PROP-E if it satisfies, $\forall i \in N, u_i(A_i) \geq \frac{1}{n}\cdot\sum_{j\in N} u_i(A_j)$.

\begin{equation*}
\begin{aligned}
w_i(A_i) &\ge \frac{1}{n}\cdot\sum_{j\in N} w_i(A_j)
\\
v_i(A_i) + v'_i(A_{-i}) - v'_i(M) &\ge \frac{1}{n}\cdot\sum_{j\in N}[v_i(A_j) + v'_i(A_{-j})  -  v'_i(M)]
\\
u_i(A_i) -v'_i(M)  &\ge \frac{1}{n}\cdot\sum_{j\in N} u_i(A_j) -v'_i(M)
\\
u_i(A_i)  &\ge \frac{1}{n}\cdot\sum_{j\in N} u_i(A_j)
\end{aligned}
\end{equation*}

Considering relaxation of PROP-E, we will prove it for PROP1-E, and the proof for PROPX-E follows in a similar fashion.
An allocation is said to be PROP1-E if it satisfies, $\exists \ k \in \{M \setminus A_i\},     u_i(A_i \cup \{k\}) \ge \frac{1}{n}\cdot\sum_{j\in N} u_i(A_j)$, i.e., item $k$ is good for agent $i$, or $\exists \ k \in A_i, u_i(A_i \setminus \{k\}) \ge \frac{1}{n}\cdot\sum_{j\in N} u_i(A_j)$, i.e., item $k$ is chore for agent $i$. 
Let allocation $A$ be PROP1-E in $\mathcal{W}$ then, $\forall i \in N, \exists k \in \{M \setminus A_i\}$, i.e., $k$ is a good for agent $i$,
\begin{equation*}
\begin{aligned}
w_i(A_i \cup \{k\}) &\ge \frac{1}{n}\cdot\sum_{j\in N} w_i(A_j)
\\
v_i(A_i \cup \{k\}) + v'_i(A_{-i} \setminus \{k\})
&\ge  
\frac{1}{n}\cdot\sum_{j\in N}[
v_i(A_j) + v'_i(A_{-j})] \\
u_i(A_i \cup \{k\})  &\ge
\frac{1}{n}\cdot\sum_{j\in N} u_i(A_j)
\end{aligned}
\end{equation*}

When $k \in A_i$, i.e., $k$ is a chore for agent $i$,

\begin{equation*}
\begin{aligned}
v_i(A_i \setminus \{k\}) + v'_i(A_{-i} \cup \{k\})
&\ge  
\frac{1}{n}\cdot\sum_{j\in N}[
v_i(A_j) + v'_i(A_{-j})] 
\end{aligned}
\end{equation*}

Next, we show the prove for MMS allocation.
An allocation is said to be MMS, if each agent gets at least its maximin share value, i.e., $\forall i \in N, u_i(A_i) \ge \mu_i$, where
$$\mu_i = \max_{(A_1,A_2,\ldots,A_n) \in \prod_n(M)} \min_{i \in N} u_i(A_i)$$

$\mathfrak{F} =$ MMS and let allocation $A$ be MMS in $\mathcal{W}$ then, $\forall i \in N$
\begin{equation*}
\begin{aligned}
v_i(A_i) + v'_i(A_{-i}) - v'_i(M) &\ge \mu_i  -  v'_i(M)
\end{aligned}
\end{equation*}

\smallskip
We now consider EQ for which $\mathfrak{T}$ does not retain. First, we define EQ and its relaxations. An allocation $A$ is said to be equitable, when $\forall i,j \in N$, $u_i(A_i) = u_j(A_j)$.
An allocation $A$ is said to be EQ1, i.e., Equitable up to one item, $u_i(A_i \setminus \{k\}) \ge u_j(A_j \setminus \{k\}), \ \exists k \in \{A_i \cup A_j\}$. An allocation $A$ is said to be EQX, i.e., Equitable up to any item, $u_i(A_i ) \ge u_j(A_j \setminus \{k\}), \ \forall k \in A_j$ and $v_{ik}\ge 0$, and $u_i(A_i \setminus \{k\}) \ge u_j(A_j), \ \forall k \in A_i$, and $v_{ik}\le 0$.
Consider the example, where $N =\{1,2\}$ and $M =\{g_1, g_2,g_3,g_4\}$. The 2-D additive valuations for agent $1$  for
$g_1: (3, -6)$, $g_2: (3, -6)$, $g_3: (1, -3)$, and $g_4: (1, -3)$.
For agent 2, the additive valuations for $g_1: (1, -8)$, $g_2: (1, -8)$, $g_3: (3, -6)$, and $g_4: (3, -6)$.
\begin{itemize}
    \item $A=\{(g_1,g_2),(g_3,g_4)\}$ is EQ in ${\mathcal{W}}$, but is not even EQ1 in ${\mathcal{V}}$. 
    \item $A=\{(g_3,g_4),(g_1,g_2)\}$ is EQ in ${\mathcal{V}}$ but is not even EQ1 in ${\mathcal{W}}$.
\end{itemize}

Thus, among fairness notions, $\mathfrak{T}$ retains EF, PROP-E, MMS and their additive relaxations. 

\smallskip
\noindent\textbf{Efficiency notions.}
We will discuss efficiency notions such as PO, MUW, MNW, and MEW.

We first consider $\mathfrak{E}=$ PO.
An allocation $A$ is Pareto Optimal (PO) if $\; \nexists \;  A'$ s.t.,
$ \forall i \in N$, $u_i(A'_i) \ge u_i(A_i)$ and $ \exists i \in N$, $u_i(A'_i) > u_i(A_i) $.
Let allocation $A$ be PO in $\mathcal{W}$, i.e., $\; \nexists \;  A'$ s.t., $ \forall i \in N$, $w_i(A'_i) \ge w_i(A_i)$ and $ \exists i \in N$, $W_i(A'_i) > W_i(A_i) $.
We can re-write that, 
$$
w_i(A'_i) + v'_i(M) \ge w_i(A_i) + v'_i(M) $$ $$ \exists i \in N, w_i(A'_i) + v'_i(M) > w_i(A_i) + v'_i(M)
$$
$A$ is PO in $\mathcal{V}$
Similarly, we can prove the reverse implication.

It is easy to verify that MUW allocation is also retained under transformation $\mathfrak{T}$.


We cannot define MNW in presence of externalities, as for goods, agents can have positive as well as negative utility. 


Note MEW is not retained using $\mathfrak{T}$.
Consider two agents $\{1,2\}$ and two goods $\{1,2\}$. Agents have additive valuations. $V_{11}=(8,-16)$, $V_{12}=(10,-15)$, $V_{21}=(5,-1)$, $V_{22}=(6,-2)$. 
$MEW(\mathcal{V}) = \{ (g_1,g_2),$ $(\emptyset)\} $, while $MEW(\mathcal{W}) =  \{(g_1), (g_2)\}$.

Among efficiency notions, we can retain PO and MUW using transformation $\mathfrak{T}$.
\end{proof}

\section{PROP-E and Average Share}
\label{sec:propeaverageshare}
We compare PROP-E and Average Share beyond 2-D. First we define valuation space in the presence of full externalities. 
The valuation function for $n$ agents is denoted by $\mathcal{V}=\{V_1,V_2,\ldots,V_n\}$. For each $i\in N$, $V_i:2^M \rightarrow \mathbb{R}^n$, i.e., $V_i \in \mathbb{R}^{n^{2^M}}$. Further, for any bundle $S\subseteq M$, $V_i(S) = (v_{i1}(S),v_{i2}(S),\ldots,v_{in}(S))$, where $v_{ij}(S)$ denotes the value agent $i$ receives when bundle $S$ is assigned to agent $j$.

\begin{proposition}
Beyond 2-D, PROP-E $\centernot\implies$ Average Share and Average Share $\centernot\implies$ PROP-E.
\end{proposition}
\begin{proof}
We will show that there is no relation between PROP-E and average share beyond 2-D, for that we will consider $n=3$.

An Allocation $A$ is said to be PROP-E, $u_i(A_i) \ge 1/n \cdot \sum_{j \in N} u_i(A_j)$. Let $i=\{1\}$
\begin{equation}
\begin{aligned}
\label{eq:prope}
u_1(A_1) &\ge  1/3 \cdot \Big[ u_1(A_1) + u_1(A_2) + u_1(A_3) \Big] \\
u_1(A_1) &\ge 1/3 \cdot   \Big[v_{11}(A_1)+v_{12}(A_2)+v_{13}(A_3) +  \\
&v_{11}(A_2)+v_{12}(A_1)+v_{13}(A_3) +  \\
&v_{11}(A_3)+v_{12}(A_2)+v_{13}(A_1) \Big]\\
\end{aligned}
\end{equation}

An Allocation $A$ is said to be average share, $u_i(A_i) \ge 1/n \cdot \sum_{k \in M}\sum_{j \in N} v_{ij}(k)$. Let $i=\{1\}$
\begin{equation}
\begin{aligned}
\label{eq:avgshare}
u_1(A_1) &\ge 1/3 \cdot  \Big[v_{11}(A_1)+v_{12}(A_1)+v_{13}(A_1) +  \\
&v_{11}(A_2)+v_{12}(A_2)+v_{13}(A_2) +  \\
&v_{11}(A_3)+v_{12}(A_3)+v_{13}(A_3) \Big]\\
\end{aligned}
\end{equation}

From Eq.~\ref{eq:prope} and ~\ref{eq:avgshare}, we conclude there is no relation between PROP-E and average share beyond 2-D.
\end{proof}


\section{Re-defining Approximate MMS}
\label{sec:alphammsotherdefinition}
 We re-write $\mu_i$ as follows, $\mu_i=\mu_i^{+} + \mu_i^{-}$ where $\mu_i^+$ corresponds to utility from assigned goods/unassigned chores and $\mu_i^-$ corresponds to utility from unassigned goods/assigned chores.
We propose two more approximate MMS definitions, such that it relaxes both utility and dis-utility obtained. The first two definition Def.~\ref{def:alphamms1} and \ref{def:alphamms2} is based on relaxing $\mu^+$ and $\mu^-$ simultaneously. Unfortunately we show that they need not exist in Lemma~\ref{prop:alphammsdef1proof} and \ref{prop:alphammsdef2proof}.
\begin{example} In order to prove this, we make few changes in the valuation profile $V^g$ of Example 3 and represent it as $V^{G}$.
We set $V^{G}_{1k_{10}}=(1000-\epsilon_1+\epsilon_3,-\epsilon_1)$, $V^{G}_{2k_{10}}=(1000-\epsilon_1+\epsilon_3,-\epsilon_1)$, and $V^{G}_{1k_{4}}=(1001-\epsilon_1+\epsilon_3,-\epsilon_1)$. We set $\forall i, V^{G}_{ik_8}=(1028-\epsilon_1+\epsilon_3,-\epsilon_1)$. We multiply 10 to $V^{G}$. We set $\epsilon_1 \le 10^{-5}$ $\epsilon_2=10^{-3}$ and $\epsilon_3=10^{-4}$ in the valuation profile $V^{G}$. We consider $\epsilon_3=10^{-4}$ so that agents have unique MMS bundle, for example, agent $1$ unique MMS bundle is $\{k_1,k_2,k_3,k_4\}$. We transform these valuations in 1-D using $\mathfrak{T}$, and the valuation profile $\mathfrak{T}(V^{G})$ is similar to the instance in \cite{procacciammsdoesntexistexample}, and it is easy to verify that MMS allocation doesn't exist. The MMS value of every agent in $\mathfrak{T}(V^{G})$ and $V^{G}$ is 40550000 and $10^4\epsilon_1$, respectively. Note that $\mu_i = \mu^+ + \mu^-$, $\mu^+_i=9\cdot10^4\epsilon_1$, and $\mu^-_i=-8\cdot10^4\epsilon_1$. Also $v'_i(M)=-40550000+10^4\epsilon_1$.
\end{example}

\begin{definition}[$\alpha$-MMS (I.)]
\label{def:alphamms1}
An allocation $A$ is said to be $\alpha$-MMS if it guarantees 
$$\forall i \in N, u_i(A_i) \ge \alpha \cdot \mu^+_i + (1+\alpha) \cdot  \mu^-_i$$ where $\alpha \in [0,1]$.
\end{definition}

\begin{definition}[$\alpha$-MMS (II.)]
\label{def:alphamms2}
An allocation $A$ is said to be $\alpha$-MMS if it guarantees 
$$\forall i \in N, u_i(A_i) \ge \alpha \cdot \mu^+_i + (1/\alpha) \cdot  \mu^-_i$$ where $\alpha \in (0,1]$.
\end{definition}

Unfortunately we cannot ensure $\alpha$-MMS according to definition~\ref{def:alphamms1} and ~\ref{def:alphamms2} in 2-D. 
\begin{lemma}
\label{prop:alphammsdef1proof} There is no $\alpha$-MMS (Def.~\ref{def:alphamms1}) for the  valuation profile $V^{G}$ for any $\alpha\in[0,1]$.
\end{lemma}
\begin{proof}
Let $\epsilon_1=10^{-5}$.
An allocation is $\alpha$-MMS for $\alpha \ge 0$ \emph{iff} $\forall i, u_i(A_i)\ge   \ge 0 \cdot \mu_i^{+} + (1+0) \mu_i^{-} = \mu_i^{-}$. Note that From Eq. 2, $u_i(A_i)\ge \mu_i^-$, \emph{iff} $w_i(A_i) \ge \mu_i^--v'_i(M)$, which gives us $w_i(A_i)\ge40550000-0.8$. For this to be true, we need $w_i(A_i) \ge 40550000$ since $\mathfrak{T}(V^{G})$ has all integral values. We know that such an allocation doesn't exist. Hence for any $\alpha \in [0,1]$, $\alpha$-MMS does not exist for $V^{G}$.
\end{proof}

\begin{lemma}
\label{prop:alphammsdef2proof}
An $\alpha$-MMS (Def.~\ref{def:alphamms2}) allocation may not exist.
\end{lemma}
\begin{proof}
Consider $\alpha=8\cdot10^4\epsilon_1$ and since $\epsilon_1\le10^{-5}$, we obtain $\forall i, u_i(A_i)\ge 0.72-1$. From Eq. 2, $u_i(A_i)\ge 0.72-1$, \emph{iff} $w_i(A_i) \ge 0.72-1+40550000-10^4\epsilon_1$.
Since $0<10^4\epsilon_1\le0.1$ and for this to be true, we need $w_i(A_i) \ge 40550000$ since $\mathfrak{T}(V^{G})$ has all integral values. Note that as $\alpha\ge8\cdot10^4$, $u_i(A_i)\ge0.72-1$, i.e., approximation guarantees strengthens.
As we decrease $\epsilon_1$; we decrease $\alpha$ which is $8\cdot10^4\epsilon_1$ and even though we weaker the approximation guarantees, $\alpha$-MMS still doesn't exist.




\end{proof} 






\end{document}